\documentclass[12pt,oneside,reqno]{amsart}
\usepackage{caption,subcaption}
\usepackage{dcolumn}
\usepackage{amsmath,amsthm,amssymb}
\usepackage{todonotes}
\usepackage{url}
\usepackage{placeins}

\usepackage{array,booktabs}

\def\be{\begin{equation}}
\def\ee{\end{equation}}

\newtheorem*{theorem*}{Theorem}
\newtheorem*{conj*}{Conjecture}

\theoremstyle{remark}
\newtheorem*{rem*}{Remark}
\newtheorem{lemma}{Lemma}

\def\ra{\rightarrow}

\usepackage{pdflscape}

\usepackage[color, final]{showkeys}
\usepackage{xcolor}

\usepackage{lmodern}
\makeatletter
\ifcase \@ptsize \relax
  \newcommand{\miniscule}{\@setfontsize\miniscule{4}{5}}
\or
  \newcommand{\miniscule}{\@setfontsize\miniscule{5}{6}}
\or
  \newcommand{\miniscule}{\@setfontsize\miniscule{5}{6}}
\fi
\makeatother

\newcounter{todocounter}

\def\ra{\rightarrow}

\begin{document}
 \title[Comment on EYM black hole]{Comment on a regular black hole\\ with Yang-Mills hair}

\author{Piotr Bizo\'n}
\address{Institute of Theoretical Physics, Jagiellonian
University, Krak\'ow, Poland}
\email{piotr.bizon@uj.edu.pl}

\author{Maciej Maliborski}
\address{University of Vienna, Faculty of Mathematics, Oskar-Morgenstern-Platz 1, 1090 Vienna, Austria}
\address{University of Vienna, Gravitational Physics, Boltzmanngasse 5, 1090 Vienna, Austria}
  \email{maciej.maliborski@univie.ac.at}

\date{\today}%
\begin{abstract}
We comment upon a black hole with Yang-Mills hair  presented in a recent preprint by Chen, Du, and Yau.
\end{abstract}
\maketitle

In a recent preprint \cite{CDY}, Chen, Du, and Yau claim that the spherically symmetric Einstein-Yang-Mills (EYM)  equations with gauge group $SU(2)$ admit a static black hole solution that is regular at the origin, has a single horizon, and is asymptotically~flat. This solution, constructed numerically as an endstate of a kind of  EYM heat flow, is said to have bounded variation and a curvature singularity at the horizon.
\vskip 0.1cm
The purpose of this short note is to point out that  Breitenlohner, Forg\'acs, and Maison (BFM) \cite{BFM} classified  all possible global behaviours of static spherically symmetric EYM solutions which are regular at the origin. In particular, they proved existence of a one-parameter family of regular  de Sitter-like bubbles (hereafter called BFM bubbles) which have an apparent horizon at some radius $r_0$.  Since the interior part of the static `black hole solution' presented in \cite{CDY} is regular at the origin, it must be a BFM bubble for $0\leq r<r_0$. It is less clear what is the exterior part because, as we will show below, no BFM bubble can be matched continuously\footnote{Smooth matching is  precluded by Penrose's singularity theorem.} at $r_0$ to an exterior asymptotically flat solution.  On the basis of very fragmentary description given in \cite{CDY}, we suspect that the exterior solution is just  Schwarzschild.  This means that the metric  and  YM potential have  finite jumps at $r_0$, in agreement with the authors of \cite{CDY} saying that their solution has bounded variation  (there is no curvature singularity at $r_0$ though, contrary to the claim in \cite{CDY}). In our view, such discontinuous configurations are physically meaningless, and so we refrain from commenting on their stability analysis, let alone  applications in cosmology as mentioned in \cite{CDY}. In the rest of this note, we elaborate on the above points following closely the exposition in~\cite{BFM}.
\vskip 0.1cm
In terms of the static spherically symmetric YM potential $W(r)$ and metric ansatz
\begin{equation}\label{metric1}
  ds^2=-e^{-2\delta(r)} A(r) dt^2 + \frac{dr^2}{A(r)}  + r^2 d\Omega^2,
\end{equation}
where $d\Omega^2$ is the round metric on the unit 2-sphere, the EYM equations take the form (with prime denoting $d/dr$)
\begin{subequations}\label{system_r}
\begin{align}
  &  (A W')' + \frac{2A}{r}\, W'^3 + \frac{W(1-W^2)}{r^2}=0,\\
 &r A'= 1-A -  2A\, W'^2 -\frac{(1-W^2)^2}{r^2},\\
&  r\delta' = -2\, W'^2.
\end{align}
\end{subequations}
This system has a one-parameter family of local  analytic solutions near the origin
\begin{equation}\label{ic}
  W(r)= -1+b r^2+\mathcal{O}(r^4),\quad A(r)\sim 1-4 b^2 r^2+\mathcal{O}(r^4),\quad \delta(r)\sim -4 b^2 r^2+\mathcal{O}(r^4),
\end{equation}
that are uniquely specified by the parameter $b$.
\vskip 0.1cm
 Using dynamical system methods, BFM analyzed how the global behaviour of these local solutions depends on $b$. Along the way, they proved existence of a growing sequence of positive values $b_n$ ($n\in \mathbb{N}$) for which the solutions  are globally regular and asymptotically flat (these solutions were discovered numerically by Bartnik and McKinnon \cite{BM} and independently proved to exist in \cite{SW}). The sequence $b_n$ accumulates at $b_{\infty}\approx 0.7074203$ and for $b>b_{\infty}$ the metric function $A(r)$ drops to zero at some~$r_0$ with $W'(r)$ diverging as $(r_0-r)^{-1/2}$. This is a coordinate singularity due to the bad behaviour of the coordinate $r$ at $A=0$. In order to desingularize the system~\eqref{system_r} BFM introduced $N=\sqrt{A}$ and  a new radial coordinate $\rho$ defined by $dr=r N d\rho$. Then, the metric \eqref{metric1} takes the form
\begin{equation}\label{metric2}
  ds^2=-e^{-2\delta(\rho)} N(\rho)^2 dt^2 + r(\rho)^2 (d\rho^2 + d\Omega^2),
\end{equation}
and the system \eqref{system_r} becomes (with overdot denoting $d/d\rho$)
\begin{subequations}\label{system}
\begin{eqnarray}
\dot{r} &=& r N,\\
\dot{W} &=& V,\\
\dot{V}&=& (2N-\kappa) V -W (1-W^2)\\
\dot{N} &=& (\kappa-N) N - \frac{2 V^2}{r^2},\\
\left(e^{-\delta} N\right)^{\dot{}} &=& (\kappa-N) e^{-\delta} N,
\end{eqnarray}
\end{subequations}
where
\begin{equation}\label{kappa}
  \kappa=\frac{1}{2N} \left(1+N^2+\frac{2V^2}{r^2}-\frac{(1-W^2)^2}{r^2}\right).
\end{equation}
The boundary conditions \eqref{ic} of regularity  at the origin  translate into the following asymptotic behaviour for  $\rho \ra -\infty$
\begin{equation}\label{ic2}
  r(\rho)\sim e^{\rho}, \quad W(\rho)\sim -1+b e^{2\rho},\quad V(\rho) \sim 2b e^{2\rho}, \quad N(\rho)\sim 1-2 b^2 e^{2\rho},\quad \delta(\rho)\sim -4b^2 e^{2\rho}.
\end{equation}
The reason of introducing the function $\kappa$ is that, as proved in \cite{BFM}, it remains bounded for all $\rho$ (for which the solution exists), in particular it stays finite when $N$ tends to zero. Since this fact is important in the subsequent analysis, we give here its proof.

\begin{lemma}{} For regular solutions of the system \eqref{system} one has $1\leq \kappa \leq 2-N$ for all $\rho$.
\end{lemma}
\begin{proof}
From \eqref{ic2} we get
$\kappa(\rho) \sim  1+ 2 b^2 e^{2\rho}$ for $\rho\ra-\infty$,
hence  $\kappa>1$ near $-\infty$. Using \eqref{system} we obtain $\dot{\kappa}=1+\frac{2 V^2}{r^2} - \kappa^2$, hence $\dot\kappa\vert_{\kappa=1}>0$ and therefore $\kappa$ cannot cross the line $\kappa=1$ from above. Thus, $\kappa\geq 1$ for all $\rho$. Next, let $\gamma=\kappa+N$. From \eqref{ic2} we get
$\gamma(\rho) \sim  2-\frac{12}{5} b^4 e^{4\rho}$ for $\rho\ra-\infty$,
hence  $\gamma<2$ near $-\infty$. Using \eqref{system} we obtain that $\dot{\gamma}=1-\gamma^2+3N(\gamma-N)$, hence $\dot\gamma\vert_{\gamma=2}=-3(1-N)^2<0$ and therefore $\gamma$ cannot cross the line $\gamma=2$ from below. Thus, $\gamma\leq 2$ for all $\rho$.
\end{proof}

BFM showed that for sufficiently large $b$ the function $N(\rho)$ crosses zero at some~$\rho_0$, while $W(\rho)$ is monotone increasing from $-1$ to some $W(\rho_0)<0$. Note that, in view of Lemma~1, there is no singularity at $\rho_0$. It follows from equation (5a) that the radius $r(\rho)$ has a maximum at $\rho_0$ which corresponds to the apparent horizon. We call  this configuration  the BFM bubble\footnote{In the context of \cite{CDY}, it is irrelevant what happens beyond $\rho_0$ but, for the sake of completeness, we mention that $r(\rho)$  drops to zero at some $\rho_1>\rho_0$, giving rise to a compact $t=\text{const}$ geometry with a curvature singularity at $\rho_1$ (so called ``bag of gold" configuration).}.
\vskip 0.1cm
Let $W_0=W(\rho_0)$ and $V_0=V(\rho_0)$. Since $\kappa$ is bounded, it follows from \eqref{kappa} that
\begin{equation}\label{constraint_r0}
 r_0^2= (1-W_0^2)^2-2V_0^2.
\end{equation}
On the other hand, for black hole solutions with a regular horizon at $\rho=\rho_h$, it follows from (5d) that if $N(\rho_h)=0$ and $\dot N(\rho_h)>0$, then
\begin{equation}\label{constraint_rh}
 r(\rho_h) >  1-W^2(\rho_h),
\end{equation}
hence  $r_0$ cannot play a role of a regular horizon.
In fact, no asymptotically flat exterior solution can be glued continuously to a BFM bubble. To see this, let us solve the system \eqref{system} backwards starting from local asymptotically flat solutions
near infinity. Such solutions behave as follows for $\rho\ra\infty$
\begin{equation}\label{asym}
  r(\rho) \sim  e^{\rho} + M, \quad N(\rho)\sim 1 - M e^{-\rho},\quad W(\rho) \sim 1 - c e^{-\rho},\quad V(\rho) \sim c e^{-\rho},
\end{equation}
where  positive parameters $M$ (mass) and $c$ uniquely specify the solution.
Let $M_1\approx 0.8286$ be the mass of the first Bartnik-McKinnon solution and assume for simplicity of the argument that $M<M_1$ (which is the regime considered in \cite{CDY} anyway). Then, using the methods of \cite{BFM}, one can show that (for any $c$)  the function $N(\rho)$ remains bounded away from zero for all $\rho$ and  has a strictly positive global minimum at some $\rho_{\text{min}}$, proving the above claim. Moreover, a scaling argument shows that as $c\ra 0$, then $N(\rho_{\text{min}})$ tends to zero, while $r(\rho_{\text{min}})\ra 2M$, yielding for $\rho\geq \rho_{\text{min}}$   the Schwarzschild exterior in the limit.

\vskip 0.2cm
Finally, let us try to reproduce the results shown in Fig.~1 of \cite{CDY} by solving numerically the system \eqref{system} with the boundary conditions \eqref{ic2}.  To this order, we choose  $b=0.775$, which gives  $N(\rho)=0$ for $\rho_0\approx 0.8765$ and
\begin{equation} \label{data}
r_0 \approx 0.8073,\quad  W_0 \approx -0.1889,\quad  V_0\approx 0.3728,\quad \lim_{\rho\ra\rho_0} N(\rho) e^{-\delta(\rho)}\approx 5.5268.
\end{equation}
Using the notation of \cite{CDY}, we get $x_0=\ln(r_0)\approx -0.21$, as in the legend of Fig.1 in~\cite{CDY} (we did not fine-tune the parameter $b$ more accurately because $x_0$ is given in \cite{CDY} only up to two decimal places). Assuming that the exterior is Schwarzschild with mass $M=r_0/2\approx 0.4036$ (hence $W\equiv 1$ for $\rho>\rho_0$), below we plot the metric functions, YM potential, mass, and `effective charge'
vs. $\ln{r}$, duplicating Fig.~1 of \cite{CDY} but one difference: at $r_0$  our plot of $W$ has a jump from $W_0$ to  $1$, while the plot of $W$ in Fig.~1 (b) appears continuous. Note that the lapse function $N e^{-\delta}$ (not plotted in \cite{CDY})  also has a jump at $r_0$, as given in \eqref{data}.

\vspace{2ex}

\begin{figure}[h]
	\includegraphics[width=0.48\textwidth]{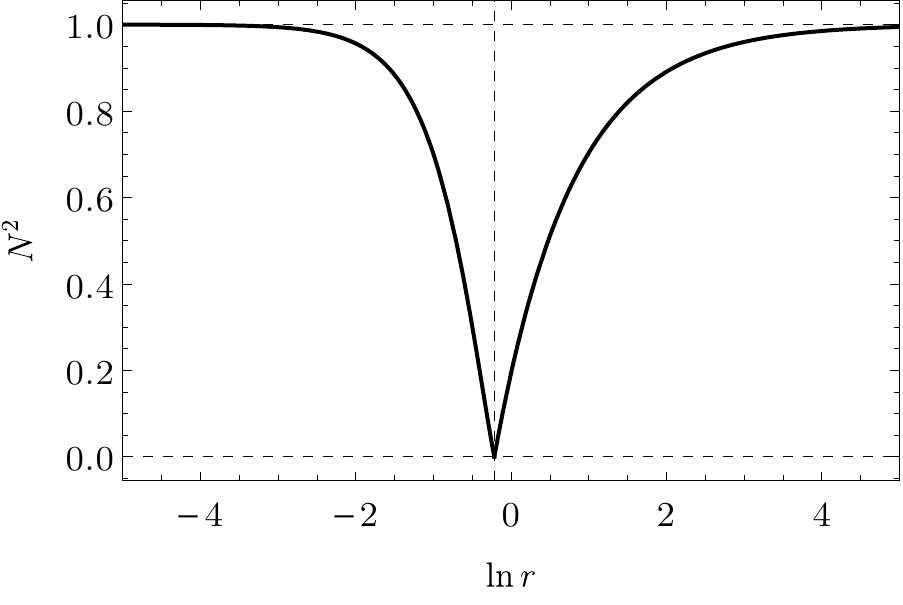}
	\includegraphics[width=0.48\textwidth]{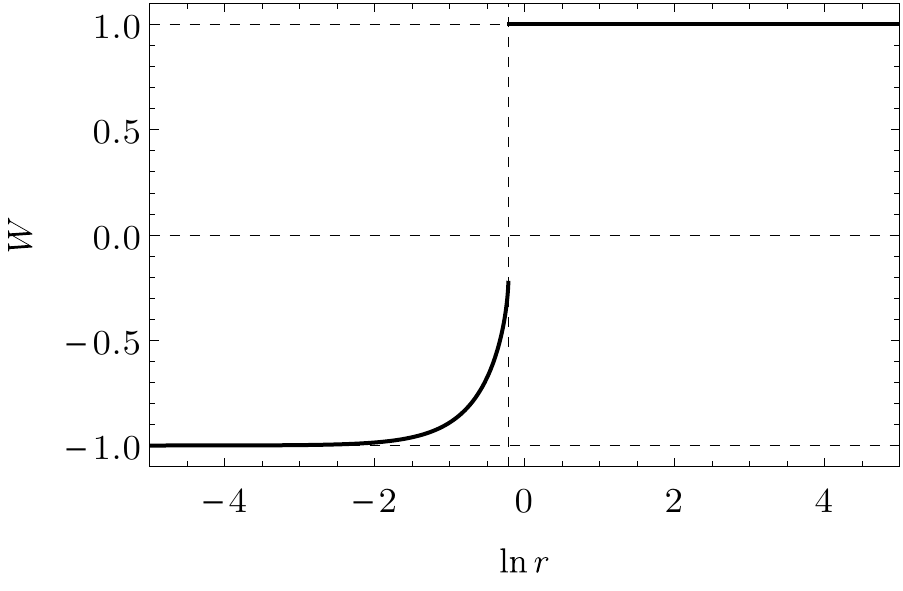}
	
	\includegraphics[width=0.48\textwidth]{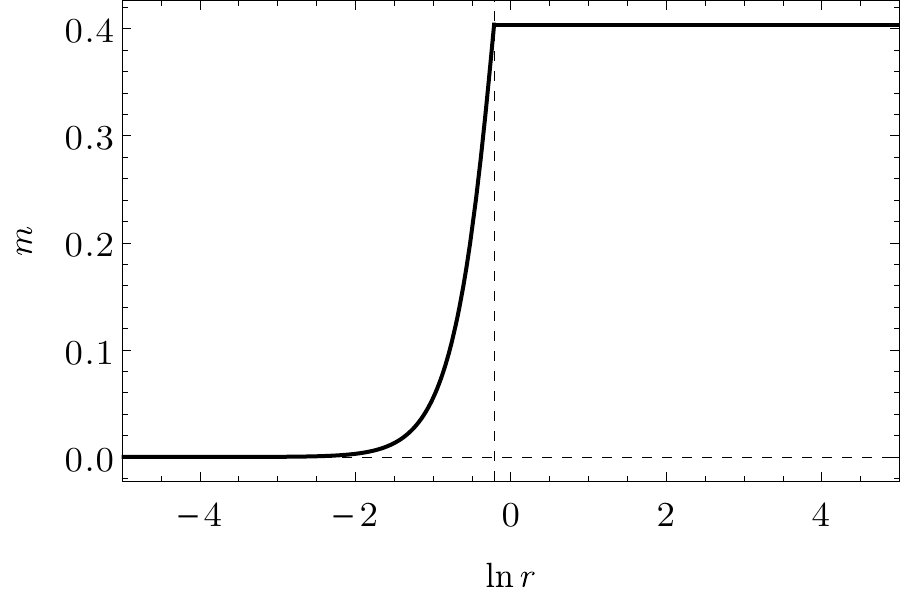}
	\includegraphics[width=0.48\textwidth]{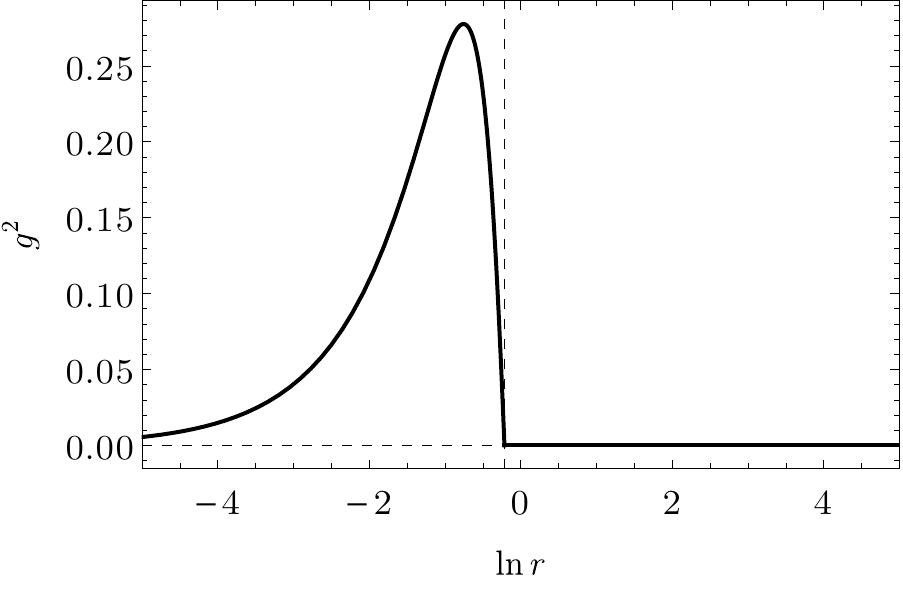}	
	\caption{\small{The BFM bubble (with the parameters given in the text) for $r<r_{0}$ matched to the exterior Schwarzschild with mass $M=r_{0}/2$ for $r>r_{0}$. The mass function and `effective charge' are defined by $m(r)=\frac{r}{2}(1-N^2(r))$ and $g^2(r)=2r(M-m(r))$, respectively. Except for the finite jump of $W(r)$ and sharp corners of other functions at $r_0$, this reproduces Fig.~1 in \cite{CDY}.}}
	\label{fig1}
\end{figure}

 The reason why the plots in \cite{CDY} are continuous and (to an eye) have no sharp corners is most likely due to the numerical method used by the authors. We suspect that their PDE solver, struggling  to satisfy a local error tolerance when the minimum of $A$ approached zero,  decreased the time step size until it became so small that the heat flow  got stuck before formation of the discontinuity. The agreement between Fig.~1 in \cite{CDY} and our Fig.~1 indicates that the resulting configuration approximately solves (away from $r_0$) the static EYM equations, however it is not clear to us why
 the heat flow proposed in \cite{CDY} should actually converge to a nontrivial steady state. Hopefully,  the supplemental materials promised in \cite{CDY} will shed some light on these issues.

\vskip 0.2cm
\noindent\emph{Acknowledgement.}
 This work was supported  by the National Science Centre grant no.\ 2017/26/A/ST2/00530.
 MM acknowledges the support of the START-Project Y963 of the Austrian Science Fund (FWF).

\end{document}